\documentclass[english, twocolumn,prl,aps]{revtex4}
\usepackage{times,graphicx,latexsym,amsfonts,amsmath,amssymb,bm,amsthm}

\theoremstyle{plain}
\newtheorem{thm}{Theorem}
\theoremstyle{plain}
\newtheorem*{thm*}{Theorem}
\theoremstyle{definition}
\newtheorem*{defn*}{Definition}
\theoremstyle{plain}
\newtheorem{lem}{Lemma}
\newtheorem*{acknowledgement}{Acknowledgement}
\newcommand{\be}{\begin{equation}}
\newcommand{\ee}{\end{equation}}
\newcommand{\ben}{\begin{eqnarray}}
\newcommand{\een}{\end{eqnarray}}

\newcommand{\nd}{\noindent}

\usepackage{babel}
\makeatother
\begin{document}

\title{Detecting q-Gaussian distributions and the normalization effect}

\author{C. Vignat and A. Plastino}
\address{$^1$L.T.H.I., E.P.F.L., Lausanne, Switzerland}

\address{$^2$Facultad de Ciencias Exactas, Universidad Nacional de La Plata and
CONICET, C.C. 727, 1900 La Plata, Argentina }
\email{vignat@univ-mlv.fr, plastino@uolsinectis.com.ar}

\begin{abstract}
\noindent We show that whenever data are gathered using a device
that performs a normalization-preprocessing, the ensuing
normalized input, as recorded by the measurement device, will
always be q-Gaussian distributed if the incoming data exhibit
elliptical symmetry. As a consequence, great care should be
exercised  when ``detecting" q-Gaussians. As an example, Gaussian
data will appear, after normalization, in the guise of q-Gaussian
records. Moreover, we show that the value of the resulting
parameter $q$ can be deduced from the normalization technique that
characterizes the device.

\vskip 0.5 mm  \noindent PACS: 05.40.-a, 05.20.Gg, 02.50.-r
\end{abstract}
\maketitle

\section{Introduction}
\nd Systems statistically described by {\it power-law probability
distributions} (PLD)
 are rather ubiquitous \cite{cero} and thus
of perennial interest \cite{uno,vign1,vign2}. Indeed, many objects
that come in different sizes have a self-similar power-law
distribution of their relative abundance over large size-ranges,
from cities to words to meteorites \cite{cero}. Now, PLDs under
variance constraint maximize a non-logarithmic information
measure, often  called Tsallis' entropy or q-entropy $H_q$
\cite{uno,vign1,vign2} \be \label{tsallis}
H_{q}\left(x\right)=\frac{1}{1-q}\,
\int_{\mathbb{R}^n}\,dx\,[f^{q}(x)-f(x)];\,\,\,\,\,\,\,\,\,q \in
\mathbb{R}. \ee This measure tends to the celebrated Shannon entropy
in the limit $q \rightarrow 1$ \cite{uno,vign1,vign2}. Systems for
whose description $H_q$ is relevant have received intense
attention in the last years, with more than 1200 papers extant and
hundreds of authors \cite{web}. To a large extent, the relevance
of the concomitant treatments relies on the fact that one often
confronts a particular scenario: {\it measuring real data
distributed according to a q-Gaussian probability law},  a special
kind of power-law probability distribution function (PDF).
 Consider a system $\mathcal{S}$ described by a vector $X$ with
$d$ components whose covariance matrix reads \be \label{covar}
K=\langle XX^{t}\rangle \equiv EXX^{t}, \ee the superscript $t$
indicating transposition. We say that $X$ is $q-$Gaussian
distributed if its probability distribution function writes in one
of  the two forms  to be found   below \cite{uno,vign1,vign2}:
\newline \nd for $1<q<\frac{d+4}{d+2}$
\begin{equation}
f_{X,q}\left(X\right)=
\frac{\Gamma\left(\frac{1}{q-1}\right)}{\Gamma\left(\frac{1}{q-1}
-\frac{d}{2}\right)\vert\pi\Lambda\vert^{1/2}}
\left(1+X^{t}\Lambda^{-1}X\right)^{\frac{1}{1-q}}.\label{eq:q>1gaussian}
\end{equation}
\nd Matrix $\Lambda$ is related to the covariance matrix $K$ according to
$\Lambda=\left(m-2\right)K,\label{eq:Kq>1}$ 
where the number of degrees  of freedom $m$ is defined \cite{vign1} in terms
 of the dimension $d$ of $X$ as
$m=\frac{2}{q-1}-d.\,\,\label{eq:mq>1}$
\nd Instead, in the case $q<1\,$ the $q-$Gaussian distribution is
\begin{equation}
 f_{X,q}\left(X\right)=\frac{\Gamma\left(\frac{2-q}{q-1}+\frac{d}{2}\right)}
{\Gamma\left(\frac{2-q}{1-q}\right) \vert\pi\Sigma\vert^{1/2}}
\left(1-X^{t}\Sigma^{-1}X\right)_{+}^{\frac{1}{1-q}},
\label{eq:q<1gaussian}
\end{equation}
where the matrix $\Sigma$ is
related to the covariance matrix via $\Sigma=pK$. We introduce
here a parameter $p$
 defined as 
$p=2\frac{2-q}{1-q}+d$ 
and use the notation $(x)_{+}=max(x,0).$
The focus of the present endeavor revolves around
the influence of the measurement device on the distribution of
these data. The many authors cited above
together with their readers should find the
considerations developed  below, concerning  the influence of the
so-called normalization stage on the performance of a measurement
device, of great interest  \cite{acknow}.

\nd Most  measurement devices consist of a {\it preprocessing
stage} that prevents the rest of the  device to be provided with
data of exceedingly large  amplitude that would cause damage to
the hardware. Since most of measured data are of stochastic
nature, the concomitant most natural and common technique is to
 \emph{statistically} normalize these input data. Since quite often
 the relevant statistical properties are of unknown character, the normalization
process consists of two steps: the data are first centered by
substraction of their estimated mean, and then scaled  by division by their
estimated standard deviation. In what follows, we detail these
operations, their statistical consequences, and the rather
surprising {\it impact} the procedure may have with regards to
non-extensive q-considerations.

\begin{figure}[h] 
\centering
\includegraphics[scale=0.6]{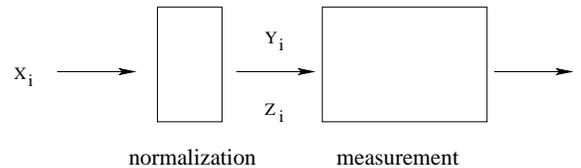}
\caption{a measurement device}
\label{fig:device}
\end{figure}

\section{The case of multivariate Gaussian data}
\subsection{Both  mean and variance unknown}
\nd  Assume that we have $n$ observations $\left\{ X_{i}\right\}
_{1\le i\le n}$ of identically distributed data, each $X_{i}$
being a vector in $\mathbb{R}^{p}$, and that neither the  mean $\mu$ nor
the covariance matrix $\Sigma$ are known. A first step of the
normalization process consists in centering the data by
substraction, from each $X_{i}$, of  an estimate $\hat{\mu}$ of
its mean devised as follows \begin{equation}
\hat{\mu}=\frac{1}{n}\sum_{j=1}^{n}X_{j}.\label{eq:mean}\end{equation}
We note that this estimate coincides with the maximum likelihood
one if the data are of Gaussian nature. Let us denote, with $1\le
i\le n,$ the residuals
\begin{equation}
\epsilon_{i}=X_{i}-\hat{\mu}
\label{epsi}
\end{equation}
so that\[ E\epsilon_{i}\equiv \langle \epsilon_{i} \rangle=0.\]

\nd  The next step is the scaling of these residuals by the
unbiased estimate of the $\left(p\times p\right)$ covariance
matrix\begin{equation} \hat{\Sigma}=\frac{1}{n}
\sum_{j=1}^{n}\epsilon_{j}\epsilon_{j}^{t}\label{eq:covariance}\end{equation}
Then the normalized version of vector $X_{i},$ denoted as $Y_{i},$
writes
\begin{equation}
Y_{i}=\hat{\Sigma}^{-\frac{1}{2}}\epsilon_{i}
=\left(\frac{1}{n}
\sum_{j=1}^{n}\epsilon_{j}\epsilon_{j}^{t}\right)^{-\frac{1}{2}}\epsilon_{i}.
\label{eq:Yi}
\end{equation}
We assume here that $n\ge p+1$ so that the estimated covariance
matrix $\hat{\Sigma}$ is positive definite with probability one \cite{eaton}. This procedure is called {\it internal
Studentization}, "internal" referring to the fact that the
estimated covariance matrix $\hat{\Sigma}$ is built upon all
available data $X_j, \,\, 1\le j\le n$, including the vector
$X_{i}$ to which it is applied. Another procedure is the {\it
so-called ``external Studentization"}: it consists in normalizing
the data $X_i$ using an estimate of the covariance matrix that
involves all data except $X_{i}$, according to \begin{eqnarray}
\hat{\Sigma}_{\left(i\right)} & = & \frac{1}{n-1}\sum_{j\ne
i}^{n}\epsilon_{j}\epsilon_{j}^{t}.\label{eq:sigmai}\end{eqnarray}
so that we denote by \begin{equation} Z_{i}=
\hat{\Sigma}_{\left(i\right)}^{-\frac{1}{2}}\epsilon_{i}=
\left(\frac{1}{n-1}\sum_{j\ne
i}^{n}\epsilon_{j}\epsilon_{j}^{t}\right)^{-\frac{1}{2}}\epsilon_{i}
\label{eq:Zi}\end{equation} the resulting normalized data. It
turns out that the distribution of normalized data $Y_{i}$ and
$Z_{i}$ can be explicitly computed when the measured data $X_{i}$ are Gaussian, as follows from the theorem below due
to Diaz-Garcia et al. \cite{diaz}.

\begin{thm}
\label{thm:1} [Diaz-Garcia] suppose that matrix
$X=\left[X_{1},\dots,X_{n}\right]$ is {\sf
Gaussian}-$\mathcal{N}\left(\mu\otimes1_{n},I_{n}\otimes\Sigma\right)-$distributed %
\footnote{Footnote 1: equivalently, vectors $X_{i}$ are independent and
identically Gaussian
$\mathcal{N}\left(\mu,\Sigma\right)$ distributed%
} with $1_{n}=\left[1,\dots,1\right]\in\mathbb{R}^{n}$. 
Then the normalized data $Y_{i}$ and $Z_{i}$ are {\sf q-Gaussian}
distributed
\begin{equation}
f_{Y_{i}}\left(Y\right)=
\frac{\Gamma
\left(\frac{n-1}{2}\right)}
{\left(\pi\left(n-1\right)\right)^{\frac{p}{2}}
\Gamma\left(\frac{n-p-1}{2}\right)}
\left(1-\frac{Y^{t}Y}{n-1}\right)_{+}^{\frac{n-p-1}{2}-1}\label{eq:fY}
\end{equation}
and
\begin{equation}
f_{Z_{i}}\left(Z\right)=
\frac{\Gamma\left(\frac{n-1}{2}\right)}
{\left(\pi\left(n-1\right)\right)^{\frac{p}{2}}
\Gamma\left(\frac{n-p-1}{2}\right)}
\left(1+\frac{Z^{t}Z}{n-1}\right)^{-\frac{n-1}{2}}.\label{eq:fZ}\end{equation}
\end{thm}
\nd We remark that in the
case of internal Studentization, the normalized data have bounded support . This corresponds to a ``hard
normalization" strategy that ensures the boundedness of the data
that feed the measurement device. On the other hand, external
Studentization corresponds to a ``soft normalization" strategy in
which boundedness of the data is not crucial to the proceedings.

\nd  Moreover, a third strategy, that may be called ``full
external Studentization", is considered in \cite{eaton}: in order
to normalize vector $X_{i}$,  the estimate of the
covariance matrix $\Sigma$ \textit{and} of the mean $\mu$ are built with
\textit{all available data but} $X_{i}$. This can be the case
whenever the measurement device builds estimates on a batch basis.
As a consequence, the
estimated mean is \[
\hat{\mu}_{\left(i\right)}=\frac{1}{n-1}\sum_{j\ne i}X_{j}\] and the estimated covariance matrix is
\[
\hat{\Sigma}_{(i)}=\frac{1}{n-1}\sum_{j\ne i}^{n}(X_j-\hat{\mu}_{\left(j\right)})(X_j-\hat{\mu}_{\left(j\right)})^t.
\]
In such a scenario  the following result holds \cite{eaton}.

\begin{thm}
\label{thm:2} [Eaton] Suppose that matrix
$X=\left[X_{1},\dots,X_{n}\right]$ is Gaussian
$\mathcal{N}\left(\mu\otimes1_{n},I_{n}\otimes\Sigma\right)-$distributed.
Then, the  random vector \begin{equation} V_{i}=
\hat{\Sigma}_{\left(i\right)}^{-\frac{1}{2}}
\left(X_{i}-\hat{\mu}_{\left(i\right)}\right)\label{eq:Vi}\end{equation}
has probability distribution function (pdf)
\begin{equation}
f_{V}\left(V\right)
=\frac{\Gamma\left(\frac{n-1}{2}\right)}
{(n\pi)^{\frac{p}{2}}
\Gamma\left(\frac{n-p-1}{2}\right)}
\left(1+\frac{V^{t}V}{n}\right)^{-\frac{n-1}{2}}.\label{eq:fV}\end{equation}
\end{thm}
\begin{proof}
See the proof in \cite{eaton}, where it is qualified as a \char`\"{}routine
multivariate calculation\char`\"{}.
\end{proof}

\subsection{Unknown variance}
In some contexts one may assume that the mean of  the data is
known. Thus, by replacing $X_i$ by $X_i-\mu$ one may state,
without loss of generality, that the mean equals zero. The two
strategies - internal or external Studentization - remain
possible, except that $\hat \mu$ should now be replaced by $0$ in
(\ref{eq:Yi}), (\ref{eq:Zi}), and (\ref{epsi}). The distributions
of the normalized data are expressed as follows.

\begin{thm}
\label{thm:3}
Under the same hypotheses as in Th.\ref{thm:1},  and assuming that
the expectation of the data vanishes, the normalized data $Y_i$
and $Z_i$ {\sf are $q-$Gaussian distributed}
\begin{equation}
f_{Y_{i}}\left(Y\right)=
\frac{\Gamma\left(\frac{n}{2}\right)}
{\left(\pi n \right)^{\frac{p}{2}}
\Gamma\left(\frac{n-p}{2}\right)}\left(1-\frac{Y^{t}Y}{n}\right)_{+}^{\frac{n-p}{2}-1}
\label{eq:fY1}\end{equation} 
and
\begin{equation}
f_{Z_{i}}\left(Z\right)= \frac{\Gamma\left(\frac{n}{2}\right)}
{\left(\pi n \right)^{\frac{p}{2}}
\Gamma\left(\frac{n-p}{2}\right)}
\left(1+\frac{Z^{t}Z}{n}\right)^{-\frac{n}{2}}.\label{eq:fZ1}
\end{equation}
\end{thm}
We remark that the known mean distributions (\ref{eq:fY1}) and (\ref{eq:fZ1}) can be recovered from the
unknown mean ones (\ref{eq:fY}) and (\ref{eq:fZ}) by replacing the parameter $n$ by $n-1$.
\section{Extension to elliptical matrix distributions}

\nd  The class of elliptically distributed random matrices plays
an important role in statistics \cite{diaz,ellenberg}. Let us
devote a few words to the concept of elliptical symmetry, {\it a
generalization} of the celebrated spherical symmetry  to which
the multi-normal distribution belongs. Spherical symmetry, 
that is invariance against rotations, found in the fundamental
laws of nature,  constitutes one of the most powerful principles
in elucidating the structure of individual atoms, complicated
molecules, entire crystals, and many other systems. Elliptical
distributions have recently  gained a lot of attention in
financial mathematics, being of use particularly in risk
management.
\nd  In what follows, we restrict our attention to the subset of
elliptical matrices which have absolutely continuous
distributions. The pertinent definition reads as follows
\cite{gupta}.
\begin{defn*}
A $\left(p\times n\right)$ random matrix $X$ has a matrix variate,
elliptical contoured distribution, denoted as
$\mathcal{E}_{p,n}\left(M,\Sigma\otimes\Phi,h\right),$ {\sf if}
its probability distribution function  writes \ben &
f_{X}\left(X\right)= \vert\Sigma\vert^{-\frac{n}{2}}
\vert\Phi\vert^{-\frac{p}{2}} \times\cr &
h\left[tr\left(\left(X-M\right)^{t}
\Sigma^{-1}\left(X-M\right)\right)\Phi^{-1}\right],\een where the
dimensions of the involved matrices are, respectively, $T:\left(p\times n\right), M:\left(p\times
n\right), \Sigma:\left(p\times p\right), \Phi:\left(n\times
n\right).$ Moreover, matrices $\Sigma$ and $\Phi$ are definite positive and function
$h:\left[0,\infty\right[\rightarrow\mathbb{R}^{+}$ is called
the density generator of $X.$ In what follows, we restrict our
attention to the case $\Phi=I_{n}$.
\end{defn*}
\nd A very important result of this paper can be phrased  as
follows:

\begin{thm}
\label{thm:4}
Theorem \ref{thm:1}, \ref{thm:2} and \ref{thm:3} still hold under the
general assumption that $X\sim\mathcal{E}_{p,n}\left(M,\Sigma\otimes I_{n},h\right).$
\end{thm}
\begin{proof}
By proper scaling, it suffices to consider $X\sim\mathcal{E}_{p,n}\left(0,I_{p,n},h\right):$
in this case, a stochastic representation of $X$ is, by lemma \eqref{lem:lemma1}
below \begin{equation}
X=rU\label{eq:XrU}\end{equation}
with \begin{equation}
U=\frac{N}{\Vert N\Vert},\label{eq:U}\end{equation}
$N$ being a Gaussian $\mathcal{N}_{n,p}\left(0,I_{n,p}\right)$ matrix.
Since any of the vectors $Y_{i}, Z_{i}$ and $V_{i}$ given by equations
\eqref{eq:Yi}, \eqref{eq:Zi} and \eqref{eq:Vi} are homogeneous
functions of order $0$ of the data $X_{i},$ we deduce that Gaussian
data can be replaced by uniform data according to \eqref{eq:U}, and
thus by elliptical data according to \eqref{eq:XrU}.
\end{proof}
\nd Note that the proof can also deduced from the more general
Thm. 5.3.1 in \cite{gupta2}.

\nd Moreover, this result can be generalized to a wider class of distributions by
noticing that  the distribution of the normalized residuals does
not depend on the covariance matrix $\Sigma$ of the data.
Consequently, matrix $\Sigma$ can be  randomly chosen within the
set of positive definite matrices, independently of the data. We
thus can state the following

\begin{thm}
\label{thm:5}
The result of theorem \ref{thm:3} extends  to the general case $X
\sim \mathcal{E}_{p,n}\left(M,\Sigma\otimes I_{n},h\right)$ with
 a random and positive definite matrix $\Sigma$.
\end{thm}

\section{Discussion and conclusions}

\nd  The following sketch illustrates the main result of this paper.

\begin{figure}[h]
\centering
\includegraphics[scale=0.5]{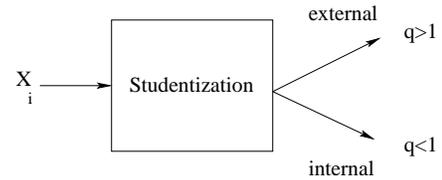}
\caption{the effect of Studentization}
\label{fig:result}
\end{figure}

\nd  In other words, any of the three scaling preprocessing
operations described above maps the set of elliptically invariant
distributions on the set of $q-$Gaussian distributed data.
Depending on the scaling method used, the resulting value of the
parameter $q$ is given as follows.

\begin{table}[h]
\begin{center}
\begin{tabular}{|c|c|c|c|}
\hline Studentization method & internal & external & full
external\tabularnewline \hline \hline value of $q$ &
$\frac{n-p-5}{n-p-3}<1$ & $\frac{n+1}{n-1}>1$ &
$\frac{n+1}{n-1}>1$\tabularnewline \hline\end{tabular}
\caption[Short Caption] {Values of parameter  $q$ according to the
normalization procedure: case where both mean and covariance are
unknown}
\end{center}
\end{table}

\begin{table}[h]
\begin{center}
\begin{tabular}{|c|c|c|}
\hline Studentization method & internal & external \tabularnewline
\hline \hline value of $q$ & $\frac{n-p-4}{n-p-4}<1$ &
$\frac{n+2}{n}>1$ \tabularnewline \hline\end{tabular}
\caption[Short Caption]{Values of  parameter $q$ according to the
normalization procedure: case where only the covariance is
unknown}
\end{center}
\end{table}

\nd  Two remarks are of interest at this point:
\begin{enumerate}
\item observation of the data after
 the  normalization stage does not allow to infer
 the distribution of the data:
 in particular, putative Gaussian data
 are systematically transformed into $q-$Gaussian data.
\item in the same way,
$q-$Gaussian data are transformed
into $q'-$Gaussian data, with
parameter $q'$ given by one of the values in Table 1, depending on
the normalization procedure: this means that the real value of
parameter $q$ is "erased" by the normalization process.
\end{enumerate}
{\bf As a conclusion}, the origin of $q-$Gaussian  data should be carefully analyzed, since they may occur, for a  very large set of recorded data (namely the set of elliptical ones), as a simple consequence of a statistical normalization step. In other words, the putative ``$q-$Gaussianity" may be a mere  artifact of the  statistical normalization step: caution is to be exercised.

As an example, a measured value of the nonextensivity parameter $q$ appears in \cite[p.230]{Europhysics} in the context of financial markets as follows:
\begin{quotation}{" \dots  returns (once demeaned and normalized by their standard deviation) have a distribution that is very well fit by q-Gaussians with $q \approx 1.4.$"}
\end{quotation}
The estimated value $q=\frac{7}{5}$ corresponds to an external  Studentization with $n=6$ data. Now, the general result from our theorem \ref{thm:5} above suggests that the "real" distribution of these data may well indeed differ from the estimated distribution.  

\section{Annex: Stochastic representation}

\begin{lem}
\label{lem:lemma1}If $X\sim\mathcal{E}_{p,n}\left(M,\Sigma\otimes I_{n},h\right)$
then the stochastic representation \[
A=M+r\Sigma^{\frac{1}{2}}U\]
 holds where $U$ is a uniform matrix on the manifold of $\left(p\times n\right)$
matrices with unit Frobenius norm (see Footnote 2
)\footnote{Footnote 2: the Frobenius norm of matrix $A$ is $\Vert
A\Vert=\sqrt{tr\left(AA^{t}\right)}$ }, and $r$ is a positive
random variable independent of matrix $U.$
\end{lem}
\begin{proof}
Associate to $A\left(n\times p\right)$ the vector
$a=vec\left(A\right)\in\mathbb{R}^{np}$ so that $\Vert
a\Vert^{2}=\Vert A\Vert^{2}$ and \[
f_{a}\left(a\right)=f\left(\Vert a\Vert^{2}\right).\] Thus a
stochastic representation for vector $a$ is \cite{Fang} $a=ru$
where $u$ is uniformly distributed on the sphere in
$\mathbb{R}^{np}$ and $r$ is positive and independent of $u;$ thus
$u$ writes\[ u=\frac{g}{\Vert g\Vert}\] where $g$ is a $(n\times
p)$ Gaussian vector; we deduce that \[ A=rU\] where $U$ is
obtained by stacking the columns of $u$ in a $\left(n\times
p\right)$ matrix form so that \[
U=\frac{G}{\sqrt{tr\left(G^{*}G\right)}}\] where $G$ is a Gaussian
$\left(n\times p\right)$ matrix, so that $U$ is uniform on the set
of $\left(n\times p\right)$ matrices with unit norm.
\end{proof}

\begin{acknowledgement}
Illuminating discussions with  Pr. A. O. Hero (E.E.C.S.,
University of Michigan) and with J. A. Diaz-Garcia (Universidad
Autonoma Agraria Antonio Narro, Mexico) are gratefully
acknowledged.
\end{acknowledgement}

\end{document}